\def\year{2019}\relax
\documentclass[letterpaper]{article} 
\usepackage{aaai19}  
\usepackage{times}  
\usepackage{helvet}  
\usepackage{courier}  
\usepackage{url}  
\usepackage{graphicx}  
\usepackage{amsmath}
\usepackage{amsthm}
\usepackage{amssymb}
\DeclareMathOperator*{\E}{\mathbb{E}}

\newcommand*{\deff}{d}
\newcommand*{\att}{a}
\DeclareMathOperator*{\argmin}{\arg\min}

\usepackage{algorithmicx}
\usepackage[noend]{algpseudocode}
\usepackage{algorithm}
\usepackage{booktabs}

\newtheorem{theorem}{Theorem}

\newtheorem{lemma}{Lemma}

\begin{document}
\title{Defending Elections Against Malicious Spread of Misinformation}
\author{Bryan Wilder$^1$ \and Yevgeniy Vorobeychik$^2$\\$^1$Center for Artificial Intelligence in Society, University of Southern California, bwilder@usc.edu\\$^2$Department of Computer Science and Engineering, Washington University in St.\ Louis, yvorobeychik@wustl.edu}
\maketitle
\begin{abstract}
	The integrity of democratic elections depends on voters' access to accurate information. However, modern media environments, which are dominated by social media, provide malicious actors with unprecedented ability to manipulate elections via misinformation, such as fake news. We study a zero-sum game between an attacker, who attempts to subvert an election by propagating a fake new story or other misinformation over a set of advertising channels, and a defender who attempts to limit the attacker's impact. Computing an equilibrium in this game is challenging as even the pure strategy sets of players are exponential. Nevertheless, we give provable polynomial-time approximation algorithms for computing the defender's minimax optimal strategy across a range of settings, encompassing different population structures as well as models of the information available to each player. Experimental results confirm that our algorithms provide near-optimal defender strategies and showcase variations in the difficulty of defending elections depending on the resources and knowledge available to the defender.
\end{abstract}

\section{Introduction}

Free and fair elections are essential to democracy. However, the integrity of elections depends on voters' access to accurate information about candidates and issues. Oftentimes, such information comes via news media or political advertising. When these information sources are accurate and transparent, they serve an important role in producing well-functioning elections. However, because of the great impact that messaging can have on voter behavior \cite{gerber2009does,dellavigna2007fox,brader2005striking}, such information can also subvert legitimate elections when deliberately falsified by malicious actors. 

In traditional media environments, such subversion is relatively difficult because professional news organizations serve as gatekeepers to information spread. However, modern media environments are increasingly decentralized due to the importance of social networks such as Facebook or Twitter, which allow outside actors to spread political information directly amongst voters \cite{chi2011twitter,wattal2010web,holcomb2013news}. This presents an unprecedented opportunity for malicious actors to spread deliberately falsified information -- ``fake news" -- and in doing so, influence the results of democratic elections. Such concerns are particularly salient in light of the 2016 U.S.\ presidential election. Recent research shows that, on average, an American adult was exposed to at least one fake news story during the campaign \cite{allcott2017social} and that these stories influenced voter attitudes \cite{pennycook2017prior}.

Prior work on election control has considered a number of mechanisms for election interference, including bribery~\cite{faliszewski2009llull,baumeister2015complexity,erdelyi2017complexity,yang2016hard},
adding or deleting voters
\cite{erdelyi2015more,loreggia2015controlling,faliszewski2011multimode,liu2009parameterized},
and adding or deleting candidates
\cite{chen2015elections,liu2009parameterized}. Only recently has social influence been explicitly studied as a means of election control \cite{sina2015adapting,wilder2018controlling,faliszewski2018opinion}. Further, with only a few exceptions which do not consider social influence~\cite{Li17,Yin18}, election control has so far primarily been studied from the attacker's perspective (to establish the computational complexity of controlling an election when the attacker is the only actor). 

We therefore ask the following natural question: \emph{how can a defender mitigate the impact of fake news on an election?} 
For instance, a social media platform or a news organization may have the ability to detect and label fake news stories on a given advertising channel, or propagate a counter-message with more accurate information. We model this interaction as a zero-sum game between an attacker, attempting to influence voters by advertising on a subset of possible channels, and a defender who enacts counter-measures on 
a subset of channels. 
The goal for the attacker is to maximize the expected number of voters who switch to the attacker's preferred candidate, whereas the defender's goal is to minimize this quantity. 
Note that in this model the defender is neutral with respect which candidate actually wins; they focus solely on minimizing the attacker's malicious influence.

Computing equilibria is computationally challenging due to the exponential number of possible actions for each player. Complicating the problem, in practice the defender may have considerable uncertainty about which candidate each voter prefers at the start of the game (information which is needed to effectively target limited resources).  We provide efficient algorithms, backed by theoretical guarantees and empirical analysis, across a range of settings:

\begin{enumerate}
	\item In the \emph{disjoint} case, each voter can be reached by only one advertising channel, modeling a case where each channel corresponds to a different demographic group. We give an FPTAS for the minimax equilibrium strategies.
	
	\item In the \emph{nondisjoint} case, each voter can be reached by an arbitrary set of channels. We first prove that the associated computational problem is APX-hard.
	We then provide an algorithm with a bicriteria guarantee: it guarantees the defender a constant-factor approximation to the optimal payoff but relaxes the budget constraint.
	
	\item We consider three models of uncertainty about voter preferences. The first is \emph{stochastic} uncertainty where the preference profile is drawn from a distribution. The second is \emph{asymmetric} uncertainty where the preference profile is drawn from a distribution and the attacker observes the realized draw. The third is \emph{adversarial} uncertainty where the preference profile is chosen to be the worst possible for the defender within an uncertainty set. Collectively, these models allow us to capture a range of assumptions about the information available to each player. Surprisingly, we show that across all three models, and in both the disjoint and nondisjoint cases, the defender can obtain exactly the same approximation ratios as when preferences are known exactly. 
\end{enumerate}

\section{Problem Formulation}

We consider a set of voters $V$ (with $|V| = n$) and a set of advertising channels $C$ (with $|C| = m$). $C$ and $V$ form a bipartite graph; that is, each voter is reachable by one or more advertising channels. The voters participate in an election between two candidates, $c_\att$ and $c_\deff$. An attacker aims to ensure that one of these candidates, $c_\att$, wins the election. A defender aims to protect the election against this manipulation. Each voter $v$ has a preferred candidate who they vote for. Let $\theta_v = 1$ if $v$ initially prefers $c_\deff$ and 0 otherwise. 

The attacker attempts to alter election results by spreading a message (a fake news story) amongst the voters. More precisely, the attacker has a limited advertising budget and can send the message through at most $k_\att$ channels. If channel $u$ is chosen by the attacker, then any voter $v$ with an edge to $u$ switches their vote to $c_\att$ with probability $p_{uv}$, where all such events are independent. The defender can protect voters from the attacker's misinformation, for example by detecting and labeling falsified stories on a given advertising channel, or by attempting to propagate a counter-message of their own. If the defender protects channel $v$, each voter connected to $v$ is ``immunized" against the attacker's message independently with probability $q_{uv}$. 
The defender may select up to $k_\deff$ channels.

We model this interaction as a zero-sum game between the attacker and defender. In this setting, equilibrium strategies are unaffected by whether one party must first commit to a strategy (formally, the Nash and Stackelberg equilibria are equivalent). Hence without loss of generality, we consider a simultaneous-move game and seek to compute a Nash equilibrium. The defender's strategy space is all subsets of $k_\deff$ channels to protect, while the attacker's strategy space consists of all subsets of $k_\att$ channels to attack. Hence, each player has an exponentially large number of pure strategies, substantially complicating equilibrium computation. 

We now introduce the attacker's objective, which determines the payoffs for the game. When the defender chooses a set of channels $S_\deff$ and the attacker chooses $S_\att$, let $f(S_\deff, S_\att)$ be the expected number of voters who previously preferred $c_\deff$ but switch their vote to $c_\att$. The randomness is over which voters are reached by the attacker's message (determined by the probabilities $p_{uv}$ and $q_{uv}$). Formally, we can express $f$ as
\begingroup\makeatletter\def\f@size{9.5}\check@mathfonts
\begin{align*}
f(S_\deff, S_\att) = \sum_{v \in V} \theta_v \left(\prod_{u \in S_\deff} 1 - q_{uv}\right) \left(1 - \prod_{u \in S_\att} 1 - p_{uv}\right)
\end{align*}
\begingroup\makeatletter\def\f@size{10}\check@mathfonts
where the first product is the probability that the defender fails to reach voter $v$ and the second is the probability that the attacker succeeds. The term $\theta_v$ means that only voters who initially prefer $c_\deff$ count (since they are the only ones who can switch). The attacker's payoff is simply $f(S_\deff, S_\att)$, while the payoff for the defender is $-f(S_\deff, S_\att)$; in words, the defender aims to minimize the spread of misinformation.

We consider two models for how the population may be structured. 
In the \emph{disjoint} model, 
the advertising channels partition the population so that each voter has an edge to exactly one channel. This models a case where the channels represent demographic groups and the attacker is deciding which demographics to target. 
In the more general \emph{nondisjoint} model, voters may be reached through multiple channels; thus, the edges can form an arbitrary bipartite graph. 

We begin by considering the case where $\theta$ (the voters' initial preferences) are common knowledge. Subsequently, we consider the setting in which voter preferences are uncertain.

\section{Related Work}

We survey related work in two areas. First, recent work in social choice studies the interaction between social influence and elections. However, all such work examines the attacker's problem of manipulating the election, leaving open the question of how elections can be defended against misinformation. Most closely related is the work of Wilder and Vorobeychik \shortcite{wilder2018controlling}, who study the attacker's problem of manipulating an election in a model where social influence spreads amongst voters from an attacker's chosen ``seed nodes". However, they do not study the corresponding defender problem. Our model is also somewhat different in that we consider advertising to voters across a set of channels, rather than influence among the voters themselves. The work of Berdereck et al.\ \shortcite{bredereck2016large} is also closely related. They study the attacker's problem in a bribery setting where a single action (e.g., placing an ad) can sway multiple voters. Faliszewski et al.\ \shortcite{faliszewski2018opinion} extend this to a domain where the initially bribed agents can influence others. Bredereck and Elkind \shortcite{bredereck2017manipulating} also study a problem of manipulating diffusions on social networks, though not specifically in the context of elections. 	\begin{algorithm}
	\caption{FPLT($\epsilon$)}\label{alg:ftpl}
	\begin{algorithmic}[1] 
		\State Arbitrarily initialize $S^0_\deff$ and $S^0_\att$
		\For{$t = t...T$}
		\State Draw $p_\att, p_\deff$ uniformly at random from $[0, \frac{1}{\epsilon}]^{m}$
		\State //TopK returns the set consisting of the indices of the smallest $k$ entries of the given vector
		\State $S_\att^t = TopK(\sum_{s  =1}^{t-1} \ell(S_\deff^s) +p_\att, k_\att)$ 
		\State $S_\deff^t = TopK(\sum_{s  =1}^{t-1} \ell(S_\att^s) + p_\deff, k_\deff)$
		\EndFor
		\State \Return $\{S_\att^t\}$ and $\{S_\deff^t\}$
	\end{algorithmic}
\end{algorithm}Sina et al.\ \shortcite{sina2015adapting} study a different form of manipulation, where edges may be added to the graph. Together, this body of work demonstrates substantial interest in the election control literature in emerging threats such as fake news. Our contribution is the first study of these problems from the perspective of a defender.  

Second, our work is related to a complementary literature on budget allocation problems. Budget allocation is the attacker's problem in our model with no defender intervention: allocating an advertising budget to maximize the number of people reached. Efficient algorithms are available for a number of variants on this model \cite{alon2012optimizing,soma2014optimal,miyauchi2015threshold,staib2017robust}. None of this work studies the game-theoretic problem of a defender trying to prevent an attacker from reaching voters. Soma et al.\ \shortcite{soma2014optimal} study a game where multiple advertisers compete for consumers, but not where one advertiser solely attempts to block the other. Their game is a potential game with pure strategy equilibria; however, it is easy to give examples in our model where the zero-sum nature of the attacker-defender interaction requires randomization. This makes equilibrium computation harder because we cannot simply use the best response dynamics.  Our work is also related to the influence blocking maximization (IBM) problem \cite{he2012influence} where one player attempts to limit the spread of a cascade in a social network. However, in IBM the starting points of the cascade are fixed in advance; in our problem the adversary chooses a randomized strategy to evade the defender.

\section{Disjoint populations}

In this setting, the population of voters is partitioned by the channels. Let $V_u$ denote the set of voters affiliated with channel $u$. Exploiting the disjoint structure of the population, we can use linearity of expectation to rewrite the utility function $f(S_\deff, S_\att)$ as

\begin{align*}
&\sum_{u\in S_\att \setminus S_\deff} \sum_{v \in V_u} \theta_v p_{uv} + \sum_{u\in S_\att \cap S_\deff}\sum_{v \in V_u} \theta_v p_{uv} (1 - q_{uv}) \\
&= \sum_{u\in S_\att} \sum_{v \in V_u} \theta_v p_{uv} - \sum_{u \in S_\att \cap S_\deff} \sum_{v \in V_u} \theta_v p_{uv}q_{uv}.
\end{align*}
%
%
Importantly, this expression is \emph{linear} in each player's decisions. More formally, let $1[S]$ denote the indicator vector of a set $S$. Define the loss vector $\ell(S_\att)$ to have value $1[u \in S_\att] \sum_{v \in V_u} \theta_v p_{uv}q_{uv}$ in coordinate $u$.    \begin{algorithm}
	\caption{OnlineGradient$\left(\eta, \alpha, T, k_\att\right)$}\label{alg:eg}
	\begin{algorithmic}[1] 
		\State $x^0_i = \frac{1}{mk_\att}$ for $i = 1...m$
		\For{$t = 1...T$}
		\State //Greedily maximizes a function subject to budget
		\State $S_\deff^t$ = Greedy($g(\cdot|x^t_\att)$, $\alpha k_\deff$)
		\State $\nabla^t = \nabla F(x^{t-1}|S_\deff^t)$
		\State $x^{t+1}$ = Update($x_t$, $\nabla^t$)
		\EndFor
		\State \Return $\{S_\deff^t\}$
		\Function {ExponentiatedGradientUpdate}{}
		\State $y^t = \min\{x^t e^{\eta \nabla^t}, 1\}$
		\State $x^{t+1} = \frac{k_\att y^t}{||y^{t}||_1}$
		\EndFunction
		\Function {EuclideanUpdate}{}
		\State $x^{t+1} = \argmin_{y \in \mathcal{X}} ||y - (x^t + \eta \nabla^t)||_2$
		\EndFunction
		
	\end{algorithmic}
\end{algorithm}Then, we have that $f(S_\deff, S_\att) =\sum_{u \in S_\att} \sum_{v \in V_u}\theta_v p_{uv} - 1[S_\deff]^\top \ell(S_\att)$, where the first term is constant with respect to $S_\deff$. Similarly, we can define a loss vector $\ell(S_\deff)$ which encapsulates the attacker's payoff for any defender action $S_\deff$.

To exploit this structure, we employ an algorithmic strategy based on online linear optimization. In such problems, a player seeks to optimize a (possibly adversarially chosen) sequence of linear functions over a feasible set. The aim is to achieve low \emph{regret}, which measures the gap in hindsight to the best fixed decision over $T$ rounds. We map online linear optimization onto our problem as follows. The feasible set for each player consists of $m$-dimensional binary vectors, where a 1 indicates that the player has chosen the corresponding channel and a 0 indicates that they have not. A vector is feasible if it sums to at most $k_\deff$ (for the defender) or $k_\att$ (for the attacker). Both the attacker and defender will choose a series of actions from the corresponding feasible sets. In iteration $t$, if the attacker chooses a set $S_\att^t$, and the defender receives a loss vector $\ell(S_\att^t)$ and suffers loss $1[S_\deff^t]\ell(S_\att^t)$. The attacker's loss functions are defined similarly.  


Each player will generate their actions using the classical \emph{Follow The Perturbed Leader (FTPL)} algorithm of Kalai and Vempala \shortcite{kalai2005efficient} (Algorithm \ref{alg:ftpl}). At each iteration, each player best responds to the uniform distribution over all strategies played so far by their opponent, plus a small random perturbation. Note that best response here corresponds to linear optimization over the player's feasible set. Since any budget-satisfying vector is feasible, we simply select the highest-weighted $k_\deff$ elements (or $k_\att$ for the attacker). Since FTPL has a no-regret guarantee for online linear optimization neither player can gain significantly by deviating from their history of play once the number of iterations is sufficiently high. More precisely, we have the following:

\begin{theorem} \label{theorem:disjoint}
	With $\frac{4n^2\max\{k_a, k_d\}}{\epsilon^2}$ iterations of FTPL, uniform distributions on $\{S_a^t\}$ and $\{S_d^t\}$ form an $\epsilon$-equilibrium. 
\end{theorem}

\section{Nondisjoint populations}

When voters may be reachable from multiple advertising channels, the approach from the previous section breaks down because utility is no longer linear for either player: selecting one channel reaches a subset of voters and hence reduces the gain from selecting additional channels. 
Indeed, we can obtain the following hardness result:

\begin{theorem}
	In the nondisjoint setting, computing an optimal defender mixed strategy is APX-hard.
\end{theorem}

The intuition is that the maximum coverage problem is essentially a special case of ours. However, diminishing returns provides useful algorithmic structure. Formally, both players' best response functions are closely related to submodular optimization problems. A set function is submodular if for all $A \subseteq B$ and $u \in V\setminus B$, $f(B \cup \{u\}) - f(B) \leq f(A \cup \{u\}) - f(A)$. We will only deal with monotone functions, where $f(A \cup \{u\}) - f(A) \geq 0$ holds for all $A, u$.  

Our overall approach is to work in the marginal space of the attacker, by keeping track of only the marginal probability that they select each channel. That is, the attacker's current mixed strategy is concisely represented by a fractional vector $x$, where $x_u$ gives the probability of selecting channel $u$. We run an approximate no-regret learning algorithm to update $x$ over a series of iterations. At each iteration $t$, $x$ is updated via a gradient step on a reward function induced by a set $S_\deff^t$ played by the defender. Specifically, we will choose $S_\deff^t$ to be an approximate best response to the current attacker mixed strategy. 

There are two principal challenges that must be solved to enable this approach. First, we need to design an appropriate no-regret algorithm for the attacker. This is a challenging task as the attacker's utility is no longer linear (or even concave) in the marginal vector $x$. Second, we need to compute approximate best responses for the defender, which itself is NP-hard.

We resolve the first challenge by running an online gradient algorithm for the attacker, where the continuous objective at each iteration is the \emph{multilinear extension} of an objective induced by the defender's strategy $S_\deff^t$. The multilinear extension is a fractional relaxation of a submodular set function. We define the multilinear extension $F(\cdot, S_\deff)$ induced by a defender strategy $S_\deff$ as 

\begingroup\makeatletter\def\f@size{9.5}\check@mathfonts
\begin{align*}
F(x|S_\deff) = \sum_{v \in V} \theta_v \left(\prod_{u \in S_\deff} 1 - q_{uv}\right) \left(1 - \prod_{u =1}^{m} 1 - x_u p_{uv}\right)
\end{align*}
\begingroup\makeatletter\def\f@size{10}\check@mathfonts
That is, $F(x|S_\deff)$ is the expected value of $f(S_\deff, S_\att)$ when each channel $u$ is independently included in $S_\att$ with probability $x_u$. This is a special case of the multilinear extension more generally defined for arbitrary submodular set functions \cite{calinescu2011maximizing}.

While $F$ is in general not concave, we show that gradient-ascent style algorithms enjoy a no-regret guarantee against a $\frac{1}{2}$-approximation of the optimal strategy in hindsight. Our general strategy is to analyze online mirror ascent for continuous submodular functions. By making specific choices for the mirror map, we obtain two concrete algorithms (the update rules in Algorithm \ref{alg:eg}). The first is standard online gradient ascent, which takes a gradient step followed by Euclidean projection onto the feasible set $\mathcal{X} = \{x| \sum_u x_u \leq k_\att, 0 \leq x \leq 1\}$. The second is an exponentiated gradient algorithm, which scales each entry of $x$ according to the gradient and then normalizes to enforce the budget constraint. We have the following convergence guarantees:

\begin{theorem} \label{theorem:oga}
	Suppose that we apply Algorithm \ref{alg:eg} to a sequence of multilinear extensions $F(\cdot|S_\deff^1)...F(\cdot|S_\deff^T)$. Let $b = \max_{|S_\deff| \leq k_\deff, u \in C}f(S_\deff, \{u\})$. Then, after $T$ iterations, we have that 
	\begin{align*}
	\frac{1}{2}\max_{x^* \in \mathcal{X}} \sum_{t = 1}^T F(x^*|S_\deff^t) - \sum_{t = 1}^T F(x^t|S_\deff^t)  \leq \sqrt{2}LD_{k_\att} \sqrt{T}.
	\end{align*}	
	where for the exponentiated gradient update, $L = b$ and $D_{k_\att} = k_\att \log(m)$ and for the Euclidean update, $L = b\sqrt{m}$ and $D_{k_\att} = \sqrt{k_{\att}}$. 
\end{theorem}

Our proof builds on the fact that for any single continuous submodular function, any local optimum is a $\frac{1}{2}$-approximation to the global optimum and translates this into the online setting.  We remark that a no-regret guarantee for online gradient ascent for submodular functions was recently shown in \cite{chen2018online}. Our more general analysis based on mirror ascent gives their result as a special case, and also allows us to analyze the exponentiated gradient update. The advantage is that the theoretical convergence rate is substantially better for exponentiated gradient, reducing the dimension dependence from $O(\sqrt{m})$ to $O(\log m)$. However, we also include the result for online gradient ascent since it tends to perform better empirically. 

The second challenge is computing defender best responses. We show that the defender's best response problem is also closely related to a submodular maximization problem. Accordingly, we can compute approximate best responses via a greedy algorithm. Specifically, we show that the defender can obtain an $\epsilon$-approximation to the optimal best response when the greedy algorithm is given an expanded budget of $\ln\left(\frac{n}{\epsilon}\right)k_\deff$ nodes.

In more detail: fix an attacker mixed strategy, denoted as $\sigma_\att$. The defender best response problem is $
\min_{|S_\deff| \leq k_\deff} \E_{S_\att \sim \sigma_\att} \left[f(S_\deff, S_\att)\right].$ That is, we wish to minimize the number of voters who switch their vote, in expectation over $\sigma_\att$. We consider the following equivalent problem 

\begin{align*}
\max_{|S_\deff| \leq k_\deff} \E_{S_\att \sim \sigma_\att} \left[f(\emptyset, S_\att) - f(S_\deff, S_\att)\right],
\end{align*}
i.e., maximizing the number of voters who do not switch as a result of the defender's action. Define $g(S_\deff|\sigma_\att) = \E_{S_\att \sim \sigma_\att} \left[f(\emptyset, S_\att) - f(S_\deff, S_\att)\right]$. The key observation enabling efficient best response computations is the following:

\begin{lemma}
	For any attacker mixed strategy $\sigma_\att$, $g(\cdot| \sigma_\att)$ is a monotone submodular function. 
\end{lemma}

Accordingly, we can compute $\epsilon$-optimal best responses by running the greedy algorithm with an expanded budget:

\begin{theorem}
	Running the greedy algorithm on the function $g$ with a budget of $\ln\left(\frac{n}{\epsilon}\right) k_\deff$ outputs a set $S_\deff$ satisfying $
	\E_{S_\att \sim \sigma_\att}\left[f(S_\deff, S_\att)\right] \leq \min_{|S^*| \leq k_\deff} \E_{S_\att \sim \sigma_\att}\left[f(S_\deff, S_\att)\right] + \epsilon.
	$
\end{theorem}

Note that running greedy with the original budget $k_\deff$ would give a $(1 - 1/e)$ approximation for the function $g$. However, a constant factor approximation for maximizing $g$ may not translate into any approximation for  minimizing $f$ because of the constant term $f(\emptyset, S_\att)$ in the definition of $g$. Expanding the budget by a logarithmic factor gives a $1 - \epsilon$ approximation with respect to $g$, and when $\epsilon$ is small enough the guarantee can be translated back in terms of $f$. 

Combining the no-regret guarantee for the attacker and the best response approximation guarantee for the defender yields the following guarantee for the sequence of sets ${S_\deff^t}$: 

%


\begin{theorem} \label{theorem:nondisjoint}
	After $T$ iterations, let $\hat{\sigma}_T$ be the uniform distribution on $S_\deff^1...S_\deff^T$ output by Algorithm \ref{alg:eg}. The defender's payoff using $\hat{\sigma}_T$ is bounded as
	
	\begin{align*}
	\max_{|S_\att| \leq k_\att} \E_{S_{\deff} \sim \hat{\sigma}_T} \left[f(S_\deff, S_\att)\right] \leq 2\left(\tau + \epsilon + \frac{\sqrt{2}LD}{\sqrt{T}}\right).
	\end{align*}
	
\end{theorem}

Now, if we take $T = \left(\frac{4\sqrt{2}LD}{\epsilon}\right)^2$ and run greedy with $\epsilon' = \frac{\epsilon}{4}$, we obtain that $\hat{\sigma}_T$ is a 2-approximation Nash equilibrium strategy for the defender up to additive loss $\epsilon$, using a budget of $\left(\ln \left(\frac{n}{\epsilon}\right) + O(1)\right)k_\deff$. Each iteration takes time $O(nm +  m\log m + mn\alpha k)$ where the first term is to compute the attacker's gradient, the second to project onto their feasible strategy set, and the third is to run greedy for the defender (see the supplement for details).

\section{Preference uncertainty}

The previous two sections showed how to compute approximately optimal equilibrium strategies for the defender when both players know the starting preferences of the voters exactly. However, in practice the preferences will be subject to uncertainty, complicating the problem of optimally targeting resources. We now explore three models of preference uncertainty, each of which makes an increasingly conservative assumption about the information available to the defender. In each case, we show how to extend our algorithmic techniques to obtain approximately optimal defender strategies. 

\subsection{Stochastic uncertainty}

We start with the least conservative assumption that the joint preference profile of the voters is drawn from a distribution which is known to both players. Each aims to maximize their payoff in expectation over the unknown draw from this distribution. We show that in both the disjoint and nondisjoint settings, the same algorithmic techniques go through with a natural modification to account for uncertainty. 


Recall that $\theta$ denotes the voter preferences. $\theta$ is now drawn from a known joint distribution $D$. Let $f_\theta(S_\deff, S_\att)$ denote the expected number of voters who switch to $c_\att$ under preferences $\theta$. The payoffs are given by $\E_{\theta\sim D}[f_\theta(S_\deff, S_\att)]$. Via linearity of expectation, we can write this as
\begin{align*}
\sum_{v \in V} \Pr[\theta_v = 1] \prod_{u \in S_\deff} \left(1 - q_{uv}\right) \left(1 - \prod_{u \in S_\att} 1 - p_{uv}\right).
\end{align*}

\begin{algorithm}
	\caption{FPLT-Asymmetric($\epsilon$)}\label{alg:ftpl-asymmetric}
	\begin{algorithmic}[1] 
		\State Arbitrarily initialize $S^0_\deff$ and $S^0_\att(\theta_j)$
		\For{$t = t...T$}
		\State Draw $p_\att^j, p_\deff$ uniformly at random from $[0, \frac{1}{\epsilon}]^{m}$
		\State //TopK returns the set consisting of the indices of the smallest $k$ entries of the given vector
		\State $S_\att^t(\theta_j) = TopK(\sum_{s  =1}^{t-1} \ell_{\theta_j}(S_\deff^s) + p_\att^j, k_\att) \, j =1...N$ 
		\State $S_\deff^t = TopK(\sum_{s  =1}^{t-1} \frac{1}{N} \sum_{j = 1}^N \ell_{\theta_j}(S_\att^s(\theta_j)) + p_\deff, k_\deff)$
		\EndFor
		\State \Return $\{S_\att^t\}$ and $\{S_\deff^t\}$
	\end{algorithmic}
\end{algorithm}

\begin{algorithm}
	\caption{OG-Asymmetric$\left(\eta, \alpha, T, k_\att N\right)$}\label{alg:ega}
	\begin{algorithmic}[1] 
		\State Draw $\theta_1...\theta_N$ iid from $D$
		\State $x^0_i(\theta_j) = \frac{1}{mk_\att}$ for $i = 1...m$, $j = 1...N$
		\For{$t = 1...T$}
		\State $S_\deff^t$ = Greedy$\left(\frac{1}{N}\sum_{j = 1}^N g(\cdot | x^{t-1}(\theta_j)), \alpha k_\deff\right)$
		\For{$j  =1...N$}
		\State $\nabla^t(\theta_j) = \nabla F(x^{t-1}(\theta_j)|S_\deff^t)$
		\State $x^{t+1}(\theta_j) =$ Update($x^t(\theta_j)$, $\nabla^t(\theta_j)$)
		\EndFor
		\EndFor
		\State \Return $\{S_\deff^t\}$
	\end{algorithmic}
\end{algorithm}

Dependence on the random preferences appears only through the term $\Pr[\theta_v = 1]$. This has two important consequences. First, we can evaluate the objective and implement the corresponding algorithms using access only to the marginals of the distribution. For many distributions of interest (e.g., product distributions where each voter adopts a preference independently), these will be known explicitly, and they can in general be evaluated to arbitrary precision via random sampling. Second, since the probability term is a nonnegative constant with respect to the strategies $S_\deff$ and $S_\att$, the payoffs retain properties such as linearity (in the disjoint case) or submodularity (in the nondisjoint case). Accordingly, we can obtain exactly the same computational guarantees as in the deterministic case, merely substituting the above expression for the payoffs:

\begin{theorem}
	By substituting $\Pr[\theta_v = 1]$ for $\theta_v$ in the definition of $f$, FTPL achieves the same guarantee for the stochastic objective as in Theorem \ref{theorem:disjoint}. Further, making this substitution in the definition of $F(x|S_\deff)$ and running Algorithm \ref{alg:eg} yields the same guarantee as in Theorem \ref{theorem:nondisjoint}.
\end{theorem}

\subsection{Asymmetric uncertainty}

We now consider a case where the true voter preferences are still drawn from a distribution, but the players have access to asymmetric information about the draw. Specifically, the defender knows only the prior distribution, while the attacker has access to the true realized draw. We aim to solve the defender problem:

\begin{align} \label{problem:asymmetric}
\min_{\sigma_\deff} \E_{\theta \sim D}\left[\max_{|S_\att| \leq k_\att } \E_{S_\deff \sim \sigma_\deff}\left[f_\theta\left(S_\att, S_\deff\right)\right]\right].
\end{align}
Here, the defender minimizes in expectation over the distribution of voter preferences, but the attacker maximizes knowing the actual draw $\theta \sim D$. We show how to compute approximately optimal defender strategies for an arbitrary distribution $D$, assuming only the ability to draw i.i.d.\ samples. We first prove a concentration bound for the number of samples required to approximate the true problem over defender mixed strategies with bounded support:

\begin{lemma} \label{lemma:concentration}
	Draw $N = O\left(\frac{n^2 m T}{\epsilon^2}\log\left(\frac{1}{\delta}\right) \log m\right)$ samples. With probability at least $1-\delta$, for defender mixed strategy $\sigma_d$ with support size at most $T$,
	
	\begin{align*}
	\Big|\E_{\theta \sim D}\Big[\max_{|S_\att| \leq k_\att } &\E_{S_\deff \sim \sigma_\deff}\left[f_\theta\left(S_\att, S_\deff\right)\right]\Big] - \\&\frac{1}{N} \sum_{i = 1}^N\max_{|S_\att| \leq k_\att } \E_{S_\deff \sim \sigma_\deff}\left[f_\theta\left(S_\att, S_\deff\right)\right]\Big| \leq \epsilon
	\end{align*}
	
\end{lemma}

We now give generalizations of our earlier algorithms for the disjoint and nondisjoint settings. Each algorithm first draws sufficient samples for Lemma \ref{lemma:concentration} to hold. Then, it simulates a separate adversary for each of the samples, mimicking the ability of the adversary to respond to the true draw of $\theta$. Each adversary runs a separate instance of a no-regret learning algorithm (FTPL for the disjoint case and online gradient for the nondisjoint case). In each iteration, the defender updates according to the \emph{expectation} over all of the adversaries (since the defender does not know the true $\theta$). More precisely, in the disjoint case, the defender's loss function in iteration $t$ is given by the average of the loss functions generated by each of the individual adversaries. The defender takes a FTPL step according to this average loss. In the nondisjoint case, the defender computes a greedy best response where the objective is given by average influence averted over all of the current adversary strategies. We show the following approximation guarantee for each setting:
\begin{theorem}
	Using inputs $T = \frac{4n^2 \max\{k_\att, k_\deff\}}{\epsilon^2}$, and $N = O\left(\frac{n^2 m T}{\epsilon^2}\log\left(\frac{1}{\delta}\right) \log m\right)$ for Algorithm \ref{alg:ftpl-asymmetric}, the uniform distribution over $\{S_\deff^t\}$ is an $\epsilon$-equilibrium defender strategy. 
\end{theorem}

\begin{theorem}
	Run Algorithm 4 with $T = \frac{2L^2D^2}{\epsilon^2}$ iterations, $\eta = \frac{1}{L\sqrt{2T}}$, $\alpha = \ln \frac{n}{\epsilon} + O(1)$, and $N =  O\left(\frac{n^3 T}{\epsilon^2}\log\left(\frac{1}{\delta}\right) \log n\right)$ samples. Let $\hat{\sigma}_T$ be the uniform distribution on $S_\deff^1...S_\deff^T$. With probability at least $1 - \delta$, the defender's payoff using $\hat{\sigma}_T$ is bounded as
	
	\begin{align*}
	\E_{\theta \sim D}\left[\max_{|S_\att| \leq k_\att } \E_{S_\deff \sim \hat{\sigma}_T}\left[f_\theta\left(S_\att, S_\deff\right)\right]\right] \leq 2\tau + \epsilon.
	\end{align*}
	
	where $\tau$ is the optimal value for Problem \ref{problem:asymmetric}.
\end{theorem}

That is, the defender can obtain the same approximation guarantee in the same number of iterations. Each iteration takes time $O(N(mn + m \log m))$ to update all of the adversaries, while the defender best response problem still requires one call to greedy as before.

\subsection{Adversarial uncertainty}

\begin{algorithm}
	\caption{FPLT-Adversarial($\epsilon$)}\label{alg:ftpl-adversarial}
	\begin{algorithmic}[1] 
		\State Arbitrarily initialize $S^0_\deff$ and $S^0_\att$
		\For{$t = t...T$}
		\State Draw $p_\att, p_\deff$ uniformly at random from $[0, \frac{1}{\epsilon}]^{m}$
		\State //TopK returns the set consisting of the indices of the smallest $k$ entries of the given vector
		\State $S_\att^t = TopK\left(\left[\sum_{s  =1}^{t-1} \ell(S_\deff^s) +p_\att\right]_{1:m}, k_\att\right) \cup$
		\State $\quad\quad\quad\quad TopK\left(\left[\sum_{s  =1}^{t-1} \ell(S_\deff^s) +p_\att\right]_{m+1:m+n}, \ell\right)$
		\State $S_\deff^t = TopK(\sum_{s  =1}^{t-1} \ell(S_\att^s) + p_\deff, k_\deff)$
		\EndFor
		\State \Return $\{S_\att^t\}$ and $\{S_\deff^t\}$
	\end{algorithmic}
\end{algorithm}

\begin{algorithm}
	\caption{OG-Adversarial$\left(\eta, \alpha, T, k_\att\right)$}\label{alg:eg-adversarial}
	\begin{algorithmic}[1] 
		\State $x^0_i = \frac{1}{mk_\att}$ for $i = 1...m+n$
		\For{$t = 1...T$}
		\State $S_\deff^t$ = Greedy($x^{t-1}$, $\alpha k_\deff$)
		\State $\nabla^t = \nabla F(x^{t-1}|S_\deff^t)$
		\State $x^{t+1}_{1:m}$ = Update($x^t_{1:m}, \nabla^t_{1:m}, k_\att$)
		\State $x^{t+1}_{m+1:m+n}$ = Update$\left(x^t_{m+1:m+n}, \nabla^t_{m+1:m+n}, \ell\right)$
		\EndFor
		\State \Return $\{S_\deff^t\}$
	\end{algorithmic}
\end{algorithm}

We now consider the most conservative uncertainty model, in which the voters' preferences are chosen adversarially within some uncertainty set. Specifically, there is a nominal preference profile $\hat{\theta}$ (e.g., $\hat{\theta}$ may be an estimate from historical data). We are guaranteed that the true $\theta$ lies within the uncertainty set $\mathcal{U}_\ell = \{\theta: |\{v: \theta_v \not= \hat{\theta}_v\}| \leq \ell\}$. That is, the true $\theta$ may differ in up to $\ell$ places from our estimate. The defender solves the robust optimization problem 

\begin{align} \label{problem:robust}
\min_{\sigma_\deff} \max_{\theta \in \mathcal{U}_\ell} \max_{S_\att \leq |k_\att|} \E_{S_{\deff} \sim \sigma_\deff} \left[f(S_\deff, S_\att)\right] 
\end{align}
which optimizes against the worst case $\theta \in \mathcal{U}_\ell$. Note that Problem \ref{problem:robust} essentially places the choice of $\theta$ under the control of the attacker (formally, we can combine the two max operations). We show that the attacker component of the algorithms when payoffs are common knowledge can be generalized to handle this expanded strategy set. Essentially, the attacker will now have two kinds of actions. First, selecting a channel for a fake news message (as before). Second, directly reaching a given voter by changing their initial preference. We equivalently simulate the second class of actions by adding a new channel $v'$ for each voter $v$. The new channel has $q_{v', v} = 0$ and $p_{v', v} = 1$. That is, the attacker always succeeds in influencing $v$ and can never be stopped by the defender. The attacker's pure strategy set now consists of all choices of $k_\deff$ normal channels and $\ell$ of the new channels. 

Our result from the disjoint case goes through essentially unchanged. Algorithm \ref{alg:ftpl-adversarial} runs FTPL for both players, as before. The only change is in the linear optimization step for the attacker, which now selects separately the top $k_\att$ regular channels and $\ell$ new channels (lines 5 and 6). We have the following guarantee:

\begin{theorem}
	Using $T = \frac{4n^2 \max\{k_\att + \ell, k_\deff\}}{\epsilon^2}$ for Algorithm \ref{alg:ftpl-adversarial}, the uniform distribution over $\{S_\deff^t\}$ is an $\epsilon$-equilibrium defender strategy. 
\end{theorem}

The main technical difference is in the nondisjoint case, where the attacker's problem now corresponds to submodular maximization over a partition matroid (since the budget constraint is now split into two categories instead of a single category as before). More general matroid constraints can complicate submodular maximization, e.g., the greedy algorithm no longer obtains the optimal approximation ratio. 
Fortunately, our use of a continuous relaxation and online gradient ascent for the attacker
can be shown to generalize without loss to arbitrary matroid constraints:

\begin{theorem}
	After $T$ iterations, let $\hat{\sigma}_T$ be the uniform distribution on $S_\deff^1...S_\deff^T$ output by Algorithm \ref{alg:eg-adversarial}. The defender's payoff using $\hat{\sigma}_T$ (with respect to Problem \ref{problem:robust}) is bounded as
	
	\begin{align*}
	\max_{|S_\att| \leq k_\att} \E_{S_{\deff} \sim \hat{\sigma}_T} \left[f(S_\deff, S_\att)\right] \leq 2\left(\tau + \epsilon + \frac{L^2 D_{k_\att + \ell}^2}{2\sqrt{T}}\right).
	\end{align*}
	
\end{theorem}

\section{Experiments}

We now examine our algorithms' empirical performance, and what the resulting values reveal about the difficulty of defending elections across different settings. We focus on the nondisjoint setting for two reasons. First it is the more general case. Second, FTPL is guaranteed to converge to an $\epsilon$-optimal strategy in the disjoint setting, while in the nondisjoint setting is important to empirically assess our algorithm's approximation quality. Our experiments use the Yahoo webscope dataset \cite{yahoowebscope}. The dataset logs bids placed by advertisers on a set of phrases. We create instances where the phrases are advertising channels and the accounts are voters. To generate each instance, we sample a random subset of 100 channels and 500 voters. Each propagation probability is drawn uniformly at random from $[0,0.2]$ for each player. Each voter's preference is also drawn uniformly at random. All results are averaged over 30 iterations.  

We start with fully known preferences and examine the approximation quality of Algorithm \ref{alg:eg}. Importantly, we do not increase the defender's budget (i.e., $\alpha = 1$). Empirically, Algorithm \ref{alg:eg} performs substantially better than its theoretical guarantee, rendering bicriteria approximation unnecessary.

We use the mixed strategies that Algorithm \ref{alg:eg} outputs to compute upper and lower bounds on the value of the game. The upper bound $b_u$ is the attacker's best response to the defender mixed strategy, while the lower bound $b_\ell$ is the defender's best response to the attacker mixed strategy. It is easy to see that the defender cannot obtain utility better than $b_\ell$, and Algorithm \ref{alg:eg}'s mixed strategy guarantees utility no worse than $b_u$. Hence, we use $\frac{b_u - b_\ell}{b_\ell}$ as an upper bound on the optimality gap. Since finding exact best responses is NP-hard, we use mixed integer programs (see the supplement). 

Table \ref{table:optgap} shows that Algorithm \ref{alg:eg} computes highly accurate defender equilibrium strategies across a range of values for $k_a$ and $k_d$. We use $T = 50$ iterations with $\eta = 0.05$. \emph{The average optimality gap is always (provably) under 6\%}. Moreover, this value is an upper bound, and the real gap may be smaller. We conclude that Algorithm \ref{alg:eg} is highly effective at computing near-optimal defender strategies. Next, Figure \ref{fig:payoffs} examines how the attacker's payoff varies as a function of $k_a$ and $k_d$. Even for large $k_d$, the defender cannot completely erase the attacker's impact (to be expected since $q_{uv} < 1$ and so the defender's message is not perfectly effective). However, the defender can obtain a large reduction in the attacker's influence when $k_a$ is high. The empirical payoffs are convex in $k_d$, meaning that the defender achieves this reduction with a moderate value of $k_d$ and sees little improvement afterwards. When $k_a$ is low, even large defender expenditures have a relatively little impact. Intuitively, it is harder for the defender to ensure an intersection between their own strategy and the attacker's when the attacker only picks a small number of channels to begin with.
\begin{table} 
	\fontsize{9}{9} \selectfont
	\begin{tabular}{l|ccc}
		\midrule
		$k_d/k_a$  & 5 & 10 & 20\\
		\midrule
		$5$ & $0.016 \pm 0.007$ & $0.016 \pm 0.010$ & $0.026 \pm 0.015$ \\
		$10$ & $0.017 \pm 0.008$ & $0.020 \pm 0.008$ & $0.037 \pm 0.017$ \\
		$20$ & $0.014 \pm 0.006$ & $0.025 \pm 0.012$ & $0.053 \pm 0.022$ \\
		\midrule
	\end{tabular}
	\caption{Upper bound on optimality gap for Algorithm \ref{alg:eg}. Average over 30 instances; $\pm$ denotes standard deviation.} \label{table:optgap}
\end{table}

\begin{figure} 
	\centering
	\includegraphics[width=1.5in]{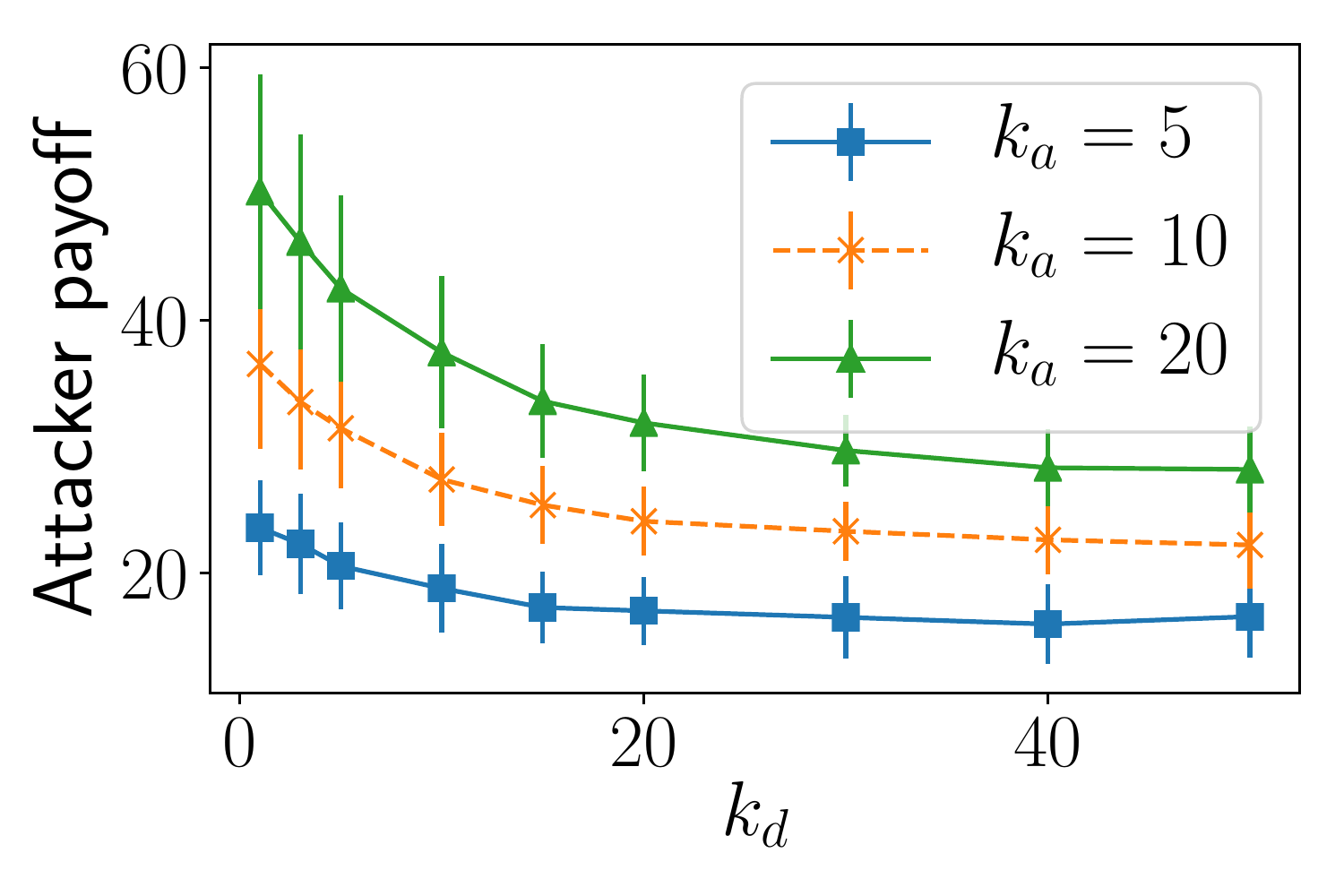}
	\includegraphics[width=1.5in]{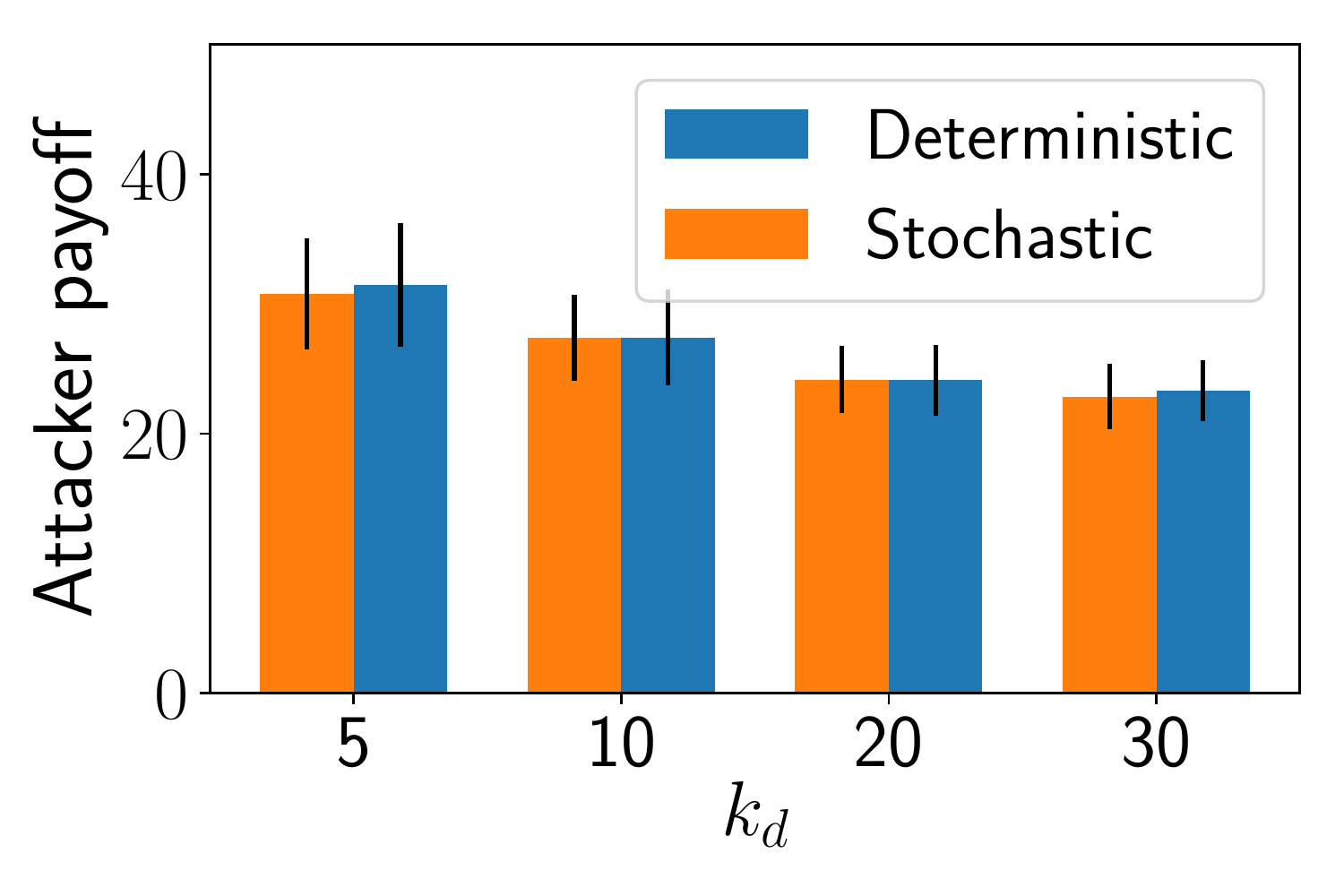}
	\includegraphics[width=1.5in]{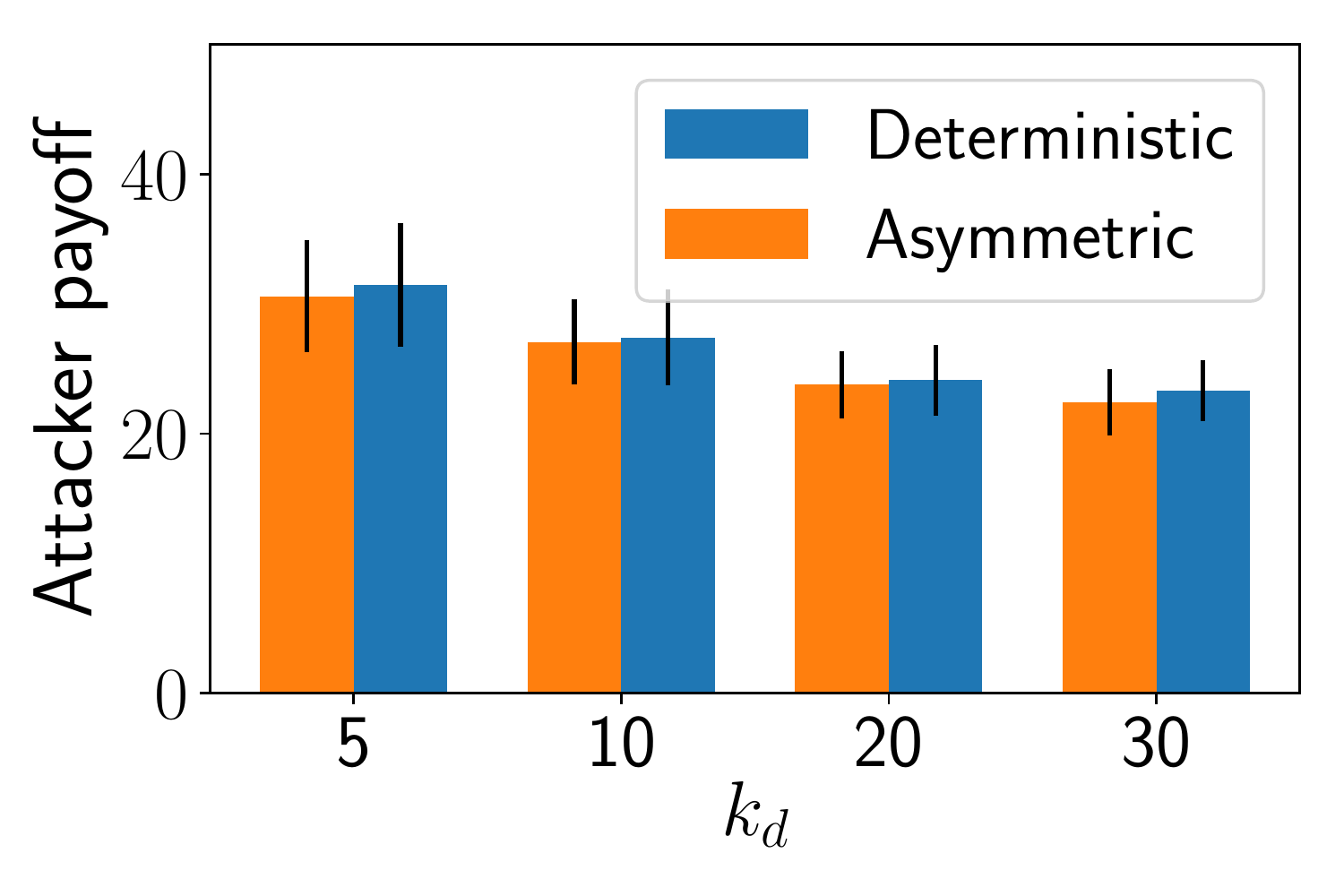}
	\includegraphics[width=1.5in]{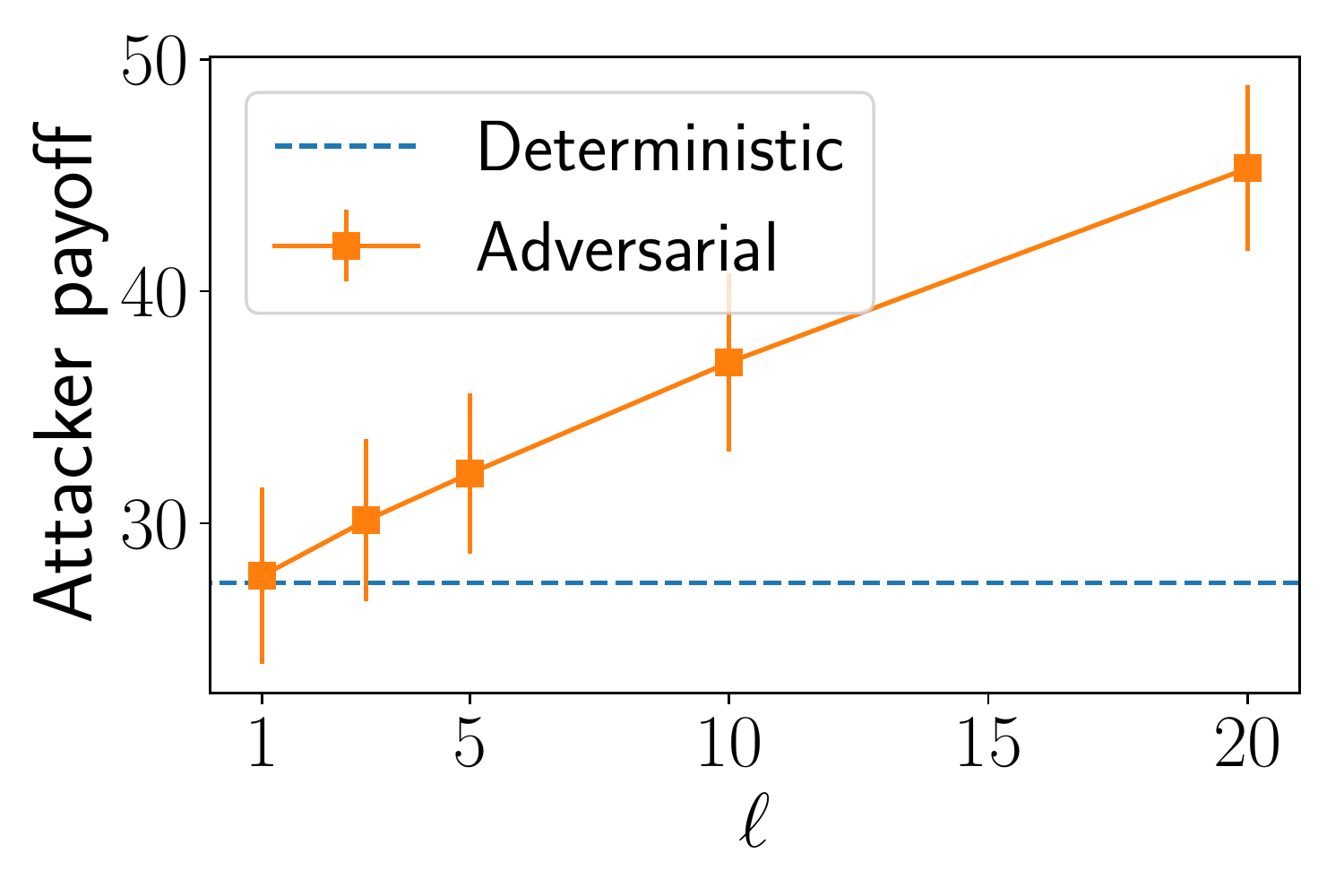}
	\caption{Top left: Attacker's payoff as the budget constraint for each player varies. Top right: attacker payoff with stochastic uncertainty. Bottom left: asymmetric uncertainty. Bottom right: adversarial uncertainty, varying the uncertainty set size $\ell$. } \label{fig:payoffs}
\end{figure}
Next, we examine the impact of uncertainty. Figure \ref{fig:payoffs} shows the attacker's payoff under stochastic, asymmetric, and adversarial uncertainty compared to fully known payoffs. Stochastic uncertainty leaves the attacker's payoff virtually identical. Surprisingly, this also holds for the asymmetric case. However, in the adversarial setting, the attacker's payoff scales linearly with $\ell$, indicating that the defender cannot mitigate the impact of such uncertainty. Hence, the defender can benefit substantially from gathering enough information to at least estimate the distribution of $\theta$, even if the attacker still has privileged information.

\textbf{Conclusion: } We introduce and study the problem of a defender mitigating the impact of adversarial misinformation on an election. Across a range of population structures and uncertainty models, we provide polynomial time approximation algorithms to compute equilibrium defender strategies, which empirically provide near-optimal payoffs. Our results show that the defender can substantially benefit from modest resource investments, and from gathering enough information to estimate voter preferences.

\textbf{Acknowledgments: } This work was partially supported by the National Science Foundation (CNS-1640624, IIS-1649972, and IIS-1526860), Office of Naval Research (N00014-15-1-2621), and Army Research Office (W911NF1610069, MURI W911NF1810208).
\small
\bibliographystyle{aaai}
\bibliography{elections_bib}

\onecolumn

\appendix
\normalsize

	\section{Hardness result}

We reduce from maximum coverage to the defender equilibrium computation problem. Suppose that we are given a family of sets $S_1...S_m$ from a universe $U$. The objective of maximum coverage is to select a subset $T$ of $k$ sets which maximize $\left| \bigcup_{S_i \in T} S_i \right|$. We create an instance of our game as follows. Each set $S$ corresponds to a channel $u_S$ and each element $i \in U$ to a voter $v_i$. Each voter has $\theta_v = 1$. Each $u_S$ has an edge to every $v_i$ such that $i \in S$. This edge has $p_{uv} = 1$ and $q_{uv} = 1$. The attacker has budget $k_a = m$ and the defender has budget $k_d = k$. Regardless of what the defender plays, an equilibrium strategy for the attacker is the pure strategy which selects all of the channels. Hence, the defender's equilibrium computation problem is identical to finding the pure strategy which maximizes the number of voters reached, since the attacker always reaches every voter, and every voter counts towards the objective since $\theta_v = 1$. This is just the maximum coverage problem. Since it is well-known that it is NP-hard to approximate maximum coverage to within a factor better than $1 - 1/e$, the theorem follows.
\section{Analysis of FTPL}
\begin{theorem}
	Let $k = \max\{k_a, k_d\}$. After $\frac{4n^2k}{\epsilon^2}$ iterations of FTPL, the uniform distribution on each player's history forms an $\epsilon$-equilibrium. 
\end{theorem}

\begin{proof}
	
	FTPL guarantees that after $T$ iterations, the defender's reward is bounded compared to the optimum as 
	
	\begin{align*}
	\sum_{t = 1}^T 1[S^t_\deff]^\top \ell(S_\att^t) - \max_{|S| \leq k_\deff}\sum_{t = 1}^T 1[S]^\top \ell(S_\att^t) \leq n\sqrt{k_d T}.
	\end{align*}
	
	By adding and subtracting the constant term in the utility function and dividing by $T$, we get 
	
	\begin{align*}
	\frac{1}{T}\sum_{t = 1}^T f(S_\deff^t, S_\att^t) - \max_{|S| \leq k_\deff} \frac{1}{T}\sum_{t = 1}^T f(S, S_\att^t) \leq \frac{n\sqrt{k_d}}{\sqrt{T}}.
	\end{align*}
	
	Applying the same reasoning from the perspective of the attacker yields
	
	\begin{align*}
	\max_{|S| \leq k_\att} \frac{1}{T}\sum_{t = 1}^T f(S_\deff^t, S) - \frac{1}{T}\sum_{t = 1}^T f(S_\deff^t, S_\att^t) \leq \frac{n\sqrt{k_a}}{\sqrt{T}}.
	\end{align*}
	
	Let $\tau$ denote the value of the game and $k = \max\{k_a, k_d\}$. We have 
	
	\begin{align*}
	\tau &= \min_{\sigma_\deff} \max_{S_\att \leq |k_\att|} \E_{S_\deff \sim \sigma_\deff}[f(S_\deff, S_\att)]\\
	&\leq \max_{S_\att \leq |k_\att|}  \E_{S_\deff \sim \hat{\sigma}_\deff}[f(S_\deff, S_\att)]\\
	&\leq \frac{1}{T}\sum_{t = 1}^T f(S_\deff^t, S_\att^t) + \frac{n\sqrt{k}}{\sqrt{T}} \quad \text{(no regret guarantee for the attacker)}
	\end{align*}
	
	This implies that 
	
	\begin{align*}
	\frac{1}{T}\sum_{t = 1}^T f(S_\deff^t, S_\att^t) \geq \tau - \frac{n\sqrt{k}}{\sqrt{T}}
	\end{align*}
	
	and so
	
	\begin{align*}
	\max_{|S| \leq k_\deff} \frac{1}{T}\sum_{t = 1}^T f(S, S_\att^t) = \max_{|S| \leq k_\deff} \E_{S_{\att} \sim \hat{\sigma}_\att}[ f(S, S_\att^t)]  \geq \tau - \frac{2n\sqrt{k}}{\sqrt{T}}.
	\end{align*}
	
	In other words, the empirical attacker strategy $\hat{\sigma}_\att$ guarantees the attacker payoff at least $\tau - \frac{2n\sqrt{k}}{\sqrt{T}}$ against \emph{any} pure strategy for the defender. This implies that $\hat{\sigma}_\att$ is a $\frac{2n\sqrt{k}}{\sqrt{T}}$-approximate equilibrium strategy for the attacker. The same line of reasoning applied to the defender completes the argument.  
	
\end{proof}

\section{Regret guarantee for online mirror ascent}

We analyze the general online mirror ascent algorithm. Our analysis draws heavily on the analysis of online mirror descent for convex functions in \cite{hazan2016introduction}, to which refer the reader for additional background. Define the Bregman divergence with respect to a function $R$ as 

\begin{align*}
B_R(x||y) = R(x) - R(y) - \nabla R(y)^\top (x-y)
\end{align*}

Let $||\cdot||_t$ be the norm induced by the Bregman divergence $B_R(x_{t}||x_{t+1})$ and $||\cdot||_t^*$ be the corresponding dual norm. Let $L$ be an upper bound on $||\nabla_t||_t^*$ and $D$ be an upper bound on $\max_{x \in \mathcal{X}} R(x) - R(x_1)$. We have the following general guarantee:

\begin{theorem}
	Let $F_1...F_T$ be a sequence of DR-submodular functions and $\nabla_t = \nabla F_t$. If we set $\eta = \frac{1}{L\sqrt{2T}}$ then
	\begin{align*}
	\frac{1}{T}\sum_{t = 1}^T F_t(x_t) \geq  \frac{1}{2}\left(\frac{1}{T}\sum_{t = 1}^T F_t(x^*)\right) -  \frac{\sqrt{2} LD}{\sqrt{T}}
	\end{align*}
	where $x^* = \max_{x \in \mathcal{X}} \sum_{t = 1}^T F_t(x^*)$.
\end{theorem}

\begin{proof}
	
	We start out by relating regret to an intermediate quantity at each step:
	\begin{lemma} \label{lemma:btl}
		$\sum_{t = 1}^T F_t(x^*) - 2F_t(x_t) \leq \sum_{t = 1}^T \nabla_t^\top (x_t - x_{t+1}) + \frac{1}{\eta} D^2$
	\end{lemma}
	
	\begin{proof}
		Define $g_0(x) = \frac{1}{\eta}R(x)$, $g_t(x) = -\nabla_t^\top x$. Via Equation 7.2 of \cite{hassani2017gradient}, we have that 
		
		\begin{align*}
		\sum_{t=1}^T F_t(x^*) - 2F_t(x_t) &\leq \sum_{t=1}^T \nabla_t^\top (x^* - x_t)\\
		&= \sum_{t=1}^T -\nabla_t^\top (x_t - x^*)\\
		&=  \sum_{t=1}^T g_t(x_t) - g_t(x^*)
		\end{align*}

		and so it suffices to bound $\sum_{t = 1}^T g_t(x_t) - g_t(x^*)$. As a first step, we show
		
		\begin{lemma}
			For any $u \in \mathcal{X}$, $\sum_{t = 0}^T g_t(u) \geq \sum_{t = 0}^T g_t(x_{t+1})$
		\end{lemma}
		\begin{proof}
			By induction on $T$. For the base case, we have that $x_1 = \argmin_{x\in \mathcal{X}} R(X)$ and so $g_0(u) \geq g_0(x_1)$. Now assume for some $T'$ that
			
			\begin{align*}
			\sum_{t = 0}^{T'} g_t(u) \geq \sum_{t = 0}^{T'} g_t(x_{t+1}).
			\end{align*}
			
			Now we will prove that the statement holds for $T'+1$. Since $x_{T'+2} = \argmin_{x \in \mathcal{X}}\sum_{t = 0}^{T'+1} g_t(x)$ we have
			
			\begin{align*}
			\sum_{t = 0}^{T'+1} g_t(u) &\geq \sum_{t = 0}^{T'+1} g_t(x_{T'+2})\\
			&= \sum_{t =0}^{T'} g_t(x_{T'+2}) + g_{T' +1}(x_{T'+2})\\
			&\geq \sum_{t = 0}^{T'} g_t(x_{t+1}) + g_{T'+1}(x_{T'+2})\\
			&= \sum_{t = 0}^{T'+1} g_t(x_{t+1}).
			\end{align*}
			
			where the third line uses the induction hypothesis for $u = x_{T'+2}$.
			
		\end{proof}
		
		Accordingly we have
		
		\begin{align*}
		\sum_{t =1}^T g_t(x_t) - g_t(x^*)&\leq \sum_{t = 1}^T \left[ g_t(x_t) - g_t(x_{t+1}) \right] + g_0(x_1) - g_0(x^*)\\
		&= \sum_{t = 1}^T g_t(x_{t}) - g_t(x_{t+1}) + \frac{1}{\eta}\left(R(x_1) - R(x^*)\right)\\
		&\leq \sum_{t = 1}^T g_t(x_{t}) - g_t(x_{t+1}) + \frac{1}{\eta} D^2
		\end{align*}
		which concludes the proof of Lemma \ref{lemma:btl}. 
	\end{proof}
	
	We now proceed to prove the main theorem. Define $\Phi_t(x) =   \sum_{s = 1}^t -\eta \nabla_s^\top x + R(x)$. Using the definition of the Bregman divergence, we have that
	
	\begin{align*}
	\Phi_t(x_t) &= \Phi_t(x_{t+1}) + (x_t - x_{t+1})^\top \nabla \Phi_t(x_{t+1}) + B_{\Phi_t}(x_t || x_{t+1})\\
	&\geq \Phi_t(x_{t+1}) + B_{\Phi_t}(x_t||x_{t+1})\\
	&= \Phi_t(x_{t+1}) + B_R(x_t||x_{t+1}).
	\end{align*}
	
	The inequality uses the fact that $x_{t+1}$ minimizes $\Phi_t$ over $\mathcal{X}$. The last equality uses the fact that the term $-\nabla_s^\top x$ is linear and doesn't affect the Bregman divergence. Thus, 
	
	\begin{align*}
	B_R(x_t||x_{t+1}) &\leq \Phi_t(x_t) - \Phi_t(x_{t+1})\\
	&= (\Phi_{t-1}(x_t) - \Phi_{t-1}(x_{t+1})) - \eta \nabla_t^\top (x_t - x_{t+1})\\
	&\leq -\eta \nabla^\top (x_t - x_{t+1})
	\end{align*}
	
	Let $||\cdot||_t$ be the norm induced by $B_R$ at the point $x_t, x_{t+1}$ and $||\cdot||_t^*$ be its dual norm. Via the Cauchy-Schwarz inequality
	
	\begin{align*}
	-\nabla_t^\top (x_{t} - x_{t+1}) &\leq ||\nabla_t||^*_t \cdot ||x_{t+1} - x_{t}||_t\\
	&= ||\nabla_t||^*_t \cdot \sqrt{2 B_R(x_t || x_{t+1})}\\
	&\leq ||\nabla_t||^*_t \cdot \sqrt{2 \eta (-\nabla_t)^\top(x_t - x_{t+1})}
	\end{align*}
	
	which implies
	
	\begin{align*}
	-\nabla_t^\top (x_t - x_{t+1}) \leq 2 \eta \left(||\nabla_t||^*_t\right)^2.
	\end{align*}
	
	Combining this with Lemma \ref{lemma:btl} and optimizing over the choice of $\eta$ now suffices to prove the theorem. 
\end{proof}

Now, the theorem in the main text is obtained by specializing the regularizer $R$ to obtain the Euclidean and exponentiated gradient updates. For the Euclidean update, we can take $R(x) = \frac{1}{2}||x- x_0||_2^2$ for any $x \in \mathcal{X}$. We thus obtain the standard Euclidean projection (see \cite{hazan2016introduction} for details). For the exponentiated gradient update, we can take $R(x) = \sum_i x_i \log x_i$. We now derive the associated projection. Writing down the Bregman divergence induced by the negative entropy, we want to solve the projection problem

\begin{align*}
x = \argmin_{\{x : ||x||_1 \leq k, 0 \leq x_i \leq 1\}} \sum_i x_i \log \left(\frac{x_i}{y_i}\right) - \sum_i x_i - \sum_i y_i
\end{align*}

This gives the Lagrangian

\begin{align*}
F(x, \lambda, \nu) = \sum_i x_i \log \left(\frac{x_i}{y_i}\right) - \sum_i x_i - \sum_i y_i + \lambda \left(\sum_i x_i - k\right) + \sum_i \nu_i \left(x_i - 1\right).
\end{align*}

At the minimizer, the KKT conditions require

\begin{align*}
&\frac{d}{dx_i}F(x, \lambda, \nu) = \log \left(\frac{x_i}{y_i}\right) + \lambda + \nu_i = 0\\
&\frac{d}{d\lambda}F(x, \lambda, \nu) = \left(\sum_i x_i\right) - 1 = 0\\
&\frac{d}{d\nu_i}F(x, \lambda, \nu) =  x_i - 1 = 0
\end{align*}

Solving for $x$ in the first equation yields $x_i = y_i e^{-(\lambda + \nu_i)}$. Now if we set 
\begin{align*}
\lambda = \ln \frac{\sum_i \min \{1, y_i\}}{k}
\end{align*}

\begin{align*}
\nu_i =   \begin{cases} 
\ln y_i& \text{if } y_i \geq 1 \\
0       & \text{otherwise } 
\end{cases}
\end{align*}

it is easy to check that complementary slackness, as well as the second and third equations (primal feasibility) are also satisfied. Hence, $(x, \lambda, \nu)$ form an optimal solution.

We remark that the bounds on $L$ and $D$ for each setting are well-known because they are the same as for mirror ascent in the offline case; see \cite{hassani2017gradient,wilder2018equilibrium} for details. 

\section{Defender best response}

We prove here that greedy best responses with an expanded budget guarantee the defender an optimal best response, up to additive error $\epsilon$. We start out with a useful characterization of the surrogate function $g$:

\begin{theorem}
	For any attacker mixed strategy $\sigma_\att$, $g$ is a monotone submodular function. 
\end{theorem}

\begin{proof}
	
	We write out the objective as 
	
	\begin{align*}
	g(S_\deff) &= \E_{S_\att \sim \sigma_\att}\left[\sum_{v \in V} \left[1 - \prod_{u \in S_\att} 1 - p_{uv}\right] - \sum_{v \in V} \left(1 - \prod_{u \in S_\att} 1 - p_{uv}\right) \prod_{u \in S_\deff} 1 - q_{uv} \right]\\ 
	&= \E_{S_\att \sim \sigma_\att}\left[\sum_{v \in V} \left(1 - \prod_{u \in S_\att} 1 - p_{uv}\right) \left(1 -  \prod_{u \in S_\deff} 1 - q_{uv} \right)\right]. 	
	\end{align*}
	
	Now, it is easy to see that for $g$ is a nonnegative linear combination of submodular functions (one for each fixed draw of $S_\att$ and $v \in V$). 
\end{proof}

\begin{theorem}
	Suppose that we run the greedy algorithm on the function $g$ with a budget of $\ln\left(\frac{n}{\epsilon}\right) k_\deff$. Then, the resulting set $S_\deff$ satisfies
	
	\begin{align*}
	\E_{S_\att \sim \sigma_\att}\left[f(S_\deff, S_\att)\right] \leq \min_{|S^*| \leq k_\deff} \E_{S_\att \sim \sigma_\att}\left[f(S_\deff, S_\att)\right] + \epsilon.
	\end{align*}
\end{theorem}

\begin{proof}
	Applying Lemma \ref{lemma:greedy} with $\ell = \ln\left(\frac{n}{\epsilon}\right)k_\deff$ yields that $g(S_\deff) \geq \left(1 - \frac{\epsilon}{n}\right) g(S^*)$. Translating this in terms of the original function $f$, we have that 
	
	\begin{align*}
	\E_{S_\att \sim \sigma_\att}\left[f(S_\deff, S_\att)\right] &= \E_{S_\att \sim \sigma_\att}\left[f(\emptyset, S_\att)\right] -  g(S_\deff)\\
	&\leq \E_{S_\att \sim \sigma_\att}\left[f(\emptyset, S_\att)\right] -  \left(1 - \frac{\epsilon}{n}\right)g(S^*)\\
	&= \E_{S_\att \sim \sigma_\att}\left[f(S^*, S_\att)\right] + \frac{\epsilon}{n}g(S^*)\\
	&\leq \E_{S_\att \sim \sigma_\att}\left[f(S^*, S_\att)\right] + \epsilon\\
	&= \min_{|S| \leq k_\deff} \E_{S_\att \sim \sigma_\att}\left[f(S, S_\att)\right] + \epsilon
	\end{align*}
	
	where the first inequality uses Lemma \ref{lemma:greedy}, the second inequality uses that $g(S^*) \leq n$, and the final equality uses that $S^*$ is an optimal solution for both maximizing $g$ and minimizing $\E\left[f(\cdot, S_\att)\right]$ (since the two problems only differ by a constant). 
	
\end{proof}

\begin{lemma}\label{lemma:greedy}
	After $\ell$ iterations, the set $S_\ell$ maintained by greedy satisfies
	\begin{align*}
	g(S_\ell) \geq \left(1 - e^{-\frac{\ell}{k_d}}\right) \max_{|S^*| \leq k_\deff} g(S^*).
	\end{align*}
\end{lemma}
\begin{proof}
	Let $v_\ell$ be the item selected in iteration $\ell$. As a consequence of submodularity, we have
	
	\begin{align*}
	g(S^* \cup S_{\ell-1}) - g(S_{\ell-1}) &\leq \sum_{v \in S^* \setminus S_{\ell-1}} g(S_{\ell-1} \cup \{v\}) - g(S_{\ell-1})\\
	&\leq |S^* \setminus S_{\ell-1}| \cdot \max_{v \in V} g(S_{\ell-1} \cup \{v\}) - g(S_{\ell-1})\\
	&\leq k_\deff g(S_{\ell-1} \cup \{v_\ell\}) - g(S_{\ell-1})\\
	&= k_\deff (g(S_{\ell}) - g(S_{\ell-1}))
	\end{align*}
	
	which implies
	
	\begin{align*}
	g(S_\ell) -g(S_{\ell-1}) \geq \frac{1}{k_\deff} g(S^* \cup S_{\ell-1}) - g(S_{\ell-1})
	\end{align*}
	
	and so 
	
	\begin{align*}
	g(S^*) - g(S_\ell) \leq g(S^* \cup S_\ell) - g(S_\ell) \leq \left(1 - \frac{1}{k_\deff}\right) \left(g(S^* \cup S_{\ell-1}) - g(S_{\ell-1})\right).
	\end{align*}
	
	Since $S_0 = \emptyset$, $g(S^* \cup S_0) - g(S_0) = g(S^*)$. Applying induction, we obtain that after $\ell$ iterations, 
	
	\begin{align*}
	g(S^*) - g(S_\ell) \leq \left(1 - \frac{1}{k_\deff}\right)^\ell g(S^*) \leq e^{-\frac{\ell}{k_\deff}} g(S^*)
	\end{align*}
	
	which proves the lemma.
\end{proof}

\section{Nondisjoint case}

We now prove the full approximation guarantee for the defender in the nondisjoint case. We start out with a simple lemma, essentially capturing that the attacker does not gain any expressive power by optimizing over the relaxed continuous space instead of distributions over their feasible set. 

\begin{lemma} \label{lemma:integrality}
	For any monotone submodular function $f$ and $x \in \mathcal{X}$, there exists $S$ with $|S| \leq k$ such that $f(S) \geq \E_{S' \sim x}[f(S')]$
\end{lemma}
\begin{proof}
	This is a simple consequence of known rounding algorithms for the multilinear extension of a monotone submodular function over matroid polytopes (e.g., swap rounding \cite{chekuri2010dependent} or pipage rounding \cite{calinescu2011maximizing}). Since such algorithms produce a random set $S$ satisfying $\E[f(S)] \geq \E_{S' \sim x}[f(S')]$, the desired set must exist by the probabilistic method. 
\end{proof}

Now we proceed to the main result:

\begin{theorem}
	After $T$ iterations, let $\hat{\sigma}_T$ be the uniform distribution on $S_\deff^1...S_\deff^T$. The defender's payoff using $\hat{\sigma}_T$ is bounded as
	
	\begin{align*}
	\max_{|S_\att| \leq k_\att} \E_{S_{\deff} \sim \hat{\sigma}_T} \left[f(S_\deff, S_\att)\right] \leq 2\left(\tau + \epsilon + \frac{\sqrt{2}LD}{\sqrt{T}}\right).
	\end{align*}
	
\end{theorem}

\begin{proof}
	We can upper bound $\tau$, the value of the game, by combining the no-regret guarantee for the attacker and the best response guarantee for the defender as follows:
	\allowdisplaybreaks
	\begin{align*} 
	\tau &= \min_{\sigma^*}\max_{S_\att \leq |k_\att|} \E_{S_{\deff} \sim \sigma^*} \left[f(S_\deff, S_\att)\right]\\
	&\geq \min_{\sigma^*} \frac{1}{T} \sum_{t = 1}^T \E_{S_\att \sim x_\att^t, S_{\deff} \sim \sigma^*} \left[f(S_\deff, S_\att)\right] \quad\quad\quad\quad\quad\quad\quad\quad \text{(Lemma \ref{lemma:integrality})}\\
	&\geq  \frac{1}{T} \sum_{t = 1}^T \min_{\sigma^*} \E_{S_\att \sim x_\att^t, S_{\deff} \sim \sigma^*} \left[f(S_\deff, S_\att)\right]\\
	&= \frac{1}{T} \sum_{t = 1}^T \min_{|S_\deff| \leq k_\deff} \E_{S_\att \sim x_\att^t} \left[f(S_\deff, S_\att)\right]\\
	&\geq \frac{1}{T} \sum_{t = 1}^T \E_{S_\att \sim x_\att^t} \left[f(S_\deff^t, S_\att)\right] - \epsilon \quad\quad\quad\quad\quad\quad\quad\quad\quad\quad \text{(best response guarantee for defender)}\\
	&= \frac{1}{T} \sum_{t = 1}^T F_t(x_\att^t) - \epsilon \quad\quad\quad\quad\quad\quad\quad\quad\quad\quad\quad\quad\quad\quad\,\, \text{(definition of the multilinear extension)}\\
	&\geq \frac{1}{2}\max_{x^* \in \mathcal{X}} \frac{1}{T} \sum_{t = 1}^T F_t(x^*) - \frac{\sqrt{2}LD}{\sqrt{T}} - \epsilon \quad\quad\quad\quad\quad\quad\quad \text{(no-regret guarantee for attacker)}\\
	&\geq \frac{1}{2}\max_{|S_\att| \leq k_\att } \frac{1}{T} \sum_{t = 1}^T F_t(1_{S_\att}) - \frac{\sqrt{2}LD}{\sqrt{T}} - \epsilon\\
	&= \frac{1}{2}\max_{|S_\att| \leq k_\att } \frac{1}{T} \sum_{t = 1}^T f(S_\deff^t, S_\att) - \frac{\sqrt{2}LD}{\sqrt{T}} - \epsilon\\
	&= \frac{1}{2}\max_{|S_\att| \leq k_\att } \E_{S_{\deff} \sim \hat{\sigma}_T}\left[f(S_\deff, S_\att)\right] - \frac{\sqrt{2}LD}{\sqrt{T}} - \epsilon.
	\end{align*}
	
	and now rearranging the terms yields
	
	\begin{align*}
	\max_{|S_\att| \leq k_\att } \E_{S_{\deff} \sim \hat{\sigma}_T}\left[f(S_\deff, S_\att)\right] \leq 2\left(\tau + \frac{\sqrt{2}LD}{\sqrt{T}} + \epsilon\right)
	\end{align*}
	
	as claimed.
\end{proof}

For the runtime, we note that iteration has three steps. First, we need to compute a gradient for the attacker. This can be done in closed form and involves a sum over $n$ terms (one for each voter). By appropriately storing intermediate products for each voter, gradient computation takes $O(mn)$ time total . Next, we need to project onto the attacker's feasible set. For the Euclidean case, this can be done in time $O(m \log m)$ \cite{karimi2017stochastic}. For the exponentiated gradient update, the computations in Algorithm 2 take time $O(m)$. Lastly, we need to run greedy for the defender. In the worst case, greedy will need to evaluate every item at every iteration, resulting in $m\alpha k$ evaluations of the function $g$. Each evaluation takes time $O(n)$, again by storing intermediate products. Combining these figures yields the runtime bound in the main paper.

\section{Asymmetric uncertainty}

\begin{lemma} \label{lemma:concentration}
	Draw $N = O\left(\frac{n^3 T}{\epsilon^2}\log\left(\frac{1}{\delta}\right) \log n\right)$ samples. With probability at least $1-\delta$, for every distribution $\sigma_d$ over the defender's pure strategy space with support size at most $T$,
	
	\begin{align*}
	\Big|\E_{\theta \sim D}\Big[\max_{|S_\att| \leq k_\att } &\E_{S_\deff \sim \sigma_\deff}\left[f_\theta\left(S_\att, S_\deff\right)\right]\Big] - \\&\frac{1}{N} \sum_{i = 1}^N\max_{|S_\att| \leq k_\att } \E_{S_\deff \sim \sigma_\deff}\left[f_\theta\left(S_\att, S_\deff\right)\right]\Big| \leq \epsilon
	\end{align*}
\end{lemma}

\begin{proof}
	We will first prove that the statement holds for any fixed $\sigma_\deff$, and then calculate the number of additional samples needed in order to take union bound over all $\sigma_deff$ with support size at most $T$. For any fixed $\sigma_d$ and $\theta$, define the random variable $Y(\sigma_\deff, \theta) = \max_{|S_\att| \leq k_\att } \E_{S_\deff \sim \sigma_\deff}\left[f_\theta\left(S_\att, S_\deff\right)\right]$. Since for any pure strategies $S_\att, S_\deff$ and any $\theta$, $f_\theta\left(S_\att, S_\deff\right) \in [0, n]$ (i.e., there are at most $n$ voters so at most $n$ can switch), we also have $Y(\sigma_\deff) \in [0, n]$. Accordingly, Hoeffding's inequality yields that with $N = O\left(\frac{n^2}{\epsilon^2}\log\frac{1}{\delta}\right)$ independent samples from $D$, we have that $\left|\E_{\theta\sim D}\left[Y(\sigma_\deff, \theta)\right] - \frac{1}{N}\sum_{i = 1}^N Y(\sigma_\deff, \theta_i)\right| \leq \epsilon$. Since there are $\binom{n}{k_\deff}$ defender pure strategies, there are at most $\binom{n}{k_\deff}^T$ distributions of support size at most $T$. Note that $\log \binom{n}{k_\deff}^T \leq O\left(T n \log n\right)$. Therefore, if we take $N = O\left(\frac{n^3 T}{\epsilon^2}\log\left(\frac{1}{\delta}\right) \log n\right)$, union bound over all $\sigma_\deff$ yields the statement in the lemma. 
\end{proof}

\begin{theorem}
	Run Algorithm 4 with $T = \frac{2L^2D^2}{\epsilon^2}$ iterations, $\eta = \frac{1}{L\sqrt{2T}}$, $\alpha = \ln \frac{n}{\epsilon} + O(1)$, and $N =  O\left(\frac{n^3 T}{\epsilon^2}\log\left(\frac{1}{\delta}\right) \log n\right)$ samples. Let $\hat{\sigma}_T$ be the uniform distribution on $S_\deff^1...S_\deff^T$. With probability at least $1 - \delta$, the defender's payoff using $\hat{\sigma}_T$ is bounded as
	
	\begin{align*}
	\E_{\theta \sim D}\left[\max_{|S_\att| \leq k_\att } \E_{S_\deff \sim \hat{\sigma}_T}\left[f_\theta\left(S_\att, S_\deff\right)\right]\right] \leq 2\tau + \epsilon.
	\end{align*}
\end{theorem}

\begin{proof}
	Let $\sigma^*$ denote an optimal defender mixed strategy. We have taken $N$ sufficiently high for Lemma \ref{lemma:concentration} to guarantee that a finite sum over the samples approximates the expectation over $\theta$ up to error $\epsilon$. We will use this fact over two classes of distributions. First all distributions of support size at most $T$, where $T$ is the number of iterations run. Second, the single distribution $\sigma^*$ (which can be included in the union bound over the first class with only a constant increase in the number of samples). Now we can bound the attacker's payoff in relation to the value of the game as
	\allowdisplaybreaks
	\begin{align*}
	&\tau =  \E_{\theta \sim D}\left[\max_{|S_\att| \leq k_\att } \E_{S_\deff \sim \sigma^*}\left[f_\theta\left(S_\att, S_\deff\right)\right]\right]\\
	&\geq \frac{1}{N} \sum_{i = 1}^N \max_{|S_\att| \leq k_\att } \E_{S_\deff \sim \sigma^*}\left[f_{\theta_i}\left(S_\att, S_\deff\right)\right] - \epsilon \quad\quad\quad\quad (\text{Lemma \ref{lemma:concentration}})\\
	&\geq  \frac{1}{N} \sum_{i = 1}^N \frac{1}{T} \sum_{t = 1}^T \E_{S_\deff \sim \sigma^*, S_\att \sim x_\att^t} \left[f_{\theta_i}\left(S_\att, S_\deff\right)\right] - \epsilon\\
	&\geq \frac{1}{T} \sum_{t = 1}^T \min_{\sigma'} \frac{1}{N} \sum_{i = 1}^N \E_{S_d \sim \sigma', S_a \sim x_a^t}\left[f_{\theta_i}\left(S_\att, S_\deff\right)\right] - \epsilon\\
	&= \frac{1}{T} \sum_{t = 1}^T \min_{|S_\deff| \leq k_\deff} \frac{1}{N} \sum_{i = 1}^N \E_{ S_\att \sim x_\att^t} \left[f_{\theta_i}\left(S_\att, S_\deff\right)\right] - \epsilon\\
	&\geq \frac{1}{T} \sum_{t = 1}^T \frac{1}{N} \sum_{i = 1}^N \E_{S_\att \sim x_\att^t} \left[f_{\theta_i}\left(S_\att, S_\deff^t\right)\right] - 2\epsilon \quad\quad\quad \text{(defender best response guarantee)}\\
	&= \frac{1}{T} \sum_{t = 1}^T \frac{1}{N} \sum_{i = 1}^N F^t_{\theta_i}(x_\att^t) - 2\epsilon\\
	&= \frac{1}{2} \frac{1}{N} \sum_{i = 1}^N \max_{x^* \in \mathcal{X}}\frac{1}{T} \sum_{t = 1}^T F^t_{\theta_i}\left(x^*\right) - 3\epsilon \quad\quad\quad\quad\,\,\,\, \text{(adversary no-regret guarantee)}\\
	&= \frac{1}{2} \frac{1}{N} \sum_{i = 1}^N \max_{|S_\att| \leq k_\att}\frac{1}{T} \sum_{t = 1}^T f_{\theta_i}\left(S_\att, S_\deff^t\right) - 3\epsilon \\
	&= \frac{1}{2} \frac{1}{N} \sum_{i = 1}^N \max_{|S_\att| \leq k_\att} \E_{S_\deff \sim \hat{\sigma_T}} \left[f_{\theta_i}\left(S_\att, S_\deff\right)\right] - 3\epsilon \\
	&\geq \frac{1}{2} \E_{\theta \sim D} \left[ \max_{|S_\att| \leq k_\att} \E_{S_\deff \sim \hat{\sigma_T}} \left[f_{\theta}\left(S_\att, S_\deff\right)\right]\right] - 3\epsilon \quad\quad (\text{Lemma \ref{lemma:concentration}}).
	\end{align*}
	
	Now, the theorem follows by applying the above argument with $\frac{\epsilon}{3}$. 
	
\end{proof}

\section{Adversarial uncertainty}

Note that our no-regret guarantee for the attacker holds with respect to an arbitrary convex set. Hence, we can replace the uniform matroid polytope with the partition matroid polytope and obtain a no-regret guarantee with respect to the new attacker action space. The only other claim specific to the constraint set is Lemma \ref{lemma:integrality}, which holds for arbitrary matroid constraints. Hence, the theorem follows by the same argument as the nondisjoint case, substituting a bound on $D$ for the enlarged constraint set. 

\section{Mixed integer programs for best response}

We describe the basic idea behind the MIPs used in the experiments to compute upper bounds on the optimality gap for our algorithms. The basic idea is to use sample average approximation to linearize the expected number of voters reached by a given strategy. We will discuss just the attacker best response; the defender best response is similar. The objective for any fixed $S_d$ is 

\begin{align*}
f(S_\deff, S_\att) = \sum_{v \in V} \theta_v \left(\prod_{u \in S_\deff} 1 - q_{uv}\right) \left(1 - \prod_{u \in S_\att} 1 - p_{uv}\right)
\end{align*}

where the term $c_v := \theta_v \left(\prod_{u \in S_\deff} 1 - q_{uv}\right)$ is a constant (with respect to $S_a$) and can be precomputed. We will have a set of $Z$ sampled scenarios (we used $Z = 200$). In each scenario $i = 1...Z$ we will maintain a set of variables $r_v^i$ denoting whether each voter has been reached. In scenario $i$, we include each edge in the graph independently with probability $p_{uv}$. Let $e_i(v)$ denote the set of channels which reach voter $v$ in scenario $i$. We will associated each channel $u$ with a binary variable $\chi_u$ denoting whether the channel is selected. Then, we can obtain the optimal attacker best response by solving the following MIP:

\begin{align*}
&\max \sum_i \sum_v c_v r_i^v\\
&r_i^v \leq \sum_{u \in e_i(v)} \chi_u \quad\forall i =1...Z, v \in V\\
&\sum_{u \in C} \chi_u \leq k_a\\
&\chi_u \in \{0,1\} \quad\forall u \in C\\
& r_i^v \in [0,1] \quad\forall i =1...Z, v \in V
\end{align*}  

The only difference in the defender case is the computation of the constant $c_v$.

\end{document}